\documentclass{article}

\usepackage{amsmath,amssymb,amsfonts,amsthm}
\usepackage{fullpage}
\usepackage{color}
\usepackage{diagrams}
\usepackage{tikz}

\newtheorem{theorem}{Theorem}[section]
\newtheorem{prop}[theorem]{Proposition}
\newtheorem{lemma}[theorem]{Lemma}
\newtheorem{corollary}[theorem]{Corollary}
\theoremstyle{remark}
\newtheorem{remark}{Remark}

\begin{document}

\title{Coxeter Group Actions on Interacting Particle Systems}

\author{Jeffrey Kuan}

\date{}

\maketitle

\abstract{We provide a conceptual proof of the color--position symmetry of colored ASEP by relating it to the actions of Coxeter groups. The group action (and hence the color--position symmetry) also applies to more general interacting particle systems, such as the colored ASEP$(q,j)$ or systems with open boundary conditions. As an application, we find the asymptotics of the expected positions of second--class particles in the ASEP$(q,j)$. 
}

\section{Introduction}
In a recent paper \cite{BorodinBufetovCP}, the authors prove a color--position symmetry for a inhomogeneous ASEP, generalizing previous results for homogeneous ASEP \cite{AAV} and TASEP \cite{angel2009}. They remark on the need for a ``conceptual understanding of why it should hold.'' Here, we show that the color--position symmetry follows through a Coxeter group action on the state space of the interacting particle system. In particular, as observed in \cite{KuanAHP}, the left action of a Coxeter group can be viewed as permuting the positions, while the right action can be viewed as permuting the colors -- color--position symmetry amounts to the commutation between the two actions. 

As a by--product of this approach, it will be shown that the colored (or multi--species) ASEP$(q,j)$ (introduced in \cite{KIMRN}, generalizing the single--single model introduced in \cite{CGRS}) also satisfies a color--position symmetry. Additionally, by using the type $B$ Coxeter group, it will be shown that the colored ASEP$(q,j)$ with open boundary conditions also satisfies color--position symmetry, and has $q$--exchangeable stationary measures. By combining the color--position symmetry with hydrodynamics of the single--species ASEP$(q,j)$, we find the expected positions of second--class particles in the ASEP$(q,j)$.

Shortly before this manuscript was completed, the author was informed of a preprint \cite{Bufetov2020} which takes a similar approach to color--position symmetry with asymptotic applications, but using Hecke algebras. 

\textbf{Acknowledgements.} The author is grateful for discussions with Alexey Bufetov and Alexei Borodin.

\section{Preliminaries}

\subsection{Coxeter Groups}
We recall some background knowledge about Coxeter groups. The statements here are from \cite{BB05} and \cite{Carter}.

A Coxeter group, or Coxeter system, consists of a pair $(W,S)$ where $W$ is a group and $S$ is a set of generators. The relations are given by $\underbrace{s_is_js_i\cdots}_{m_{ij} \text{ terms}} = \underbrace{s_js_is_j \cdots}_{m_{ij} \text{ terms}} $ for any $s_i,s_j\in S$, where $m_{ij}$ are some positive integers. In this paper, we will consider two examples of Coxeter groups. The first is the  $A_{N-1}$ Coxeter group, which is defined by setting $S=\{s_1,\ldots,s_{N-1}\}$ and defining
$$
m_{ij}
= 
\begin{cases}
2, & \text{ if } | i - j | >1,\\
3, & \text{ if } | i-j | =1,\\
1, & \text{ if } i=j.
\end{cases}
$$
Concretely, the $A_{N-1}$ Coxeter group is isomorphic to the symmetric group $S_N$ on $N$ letters, and $s_i$ is the transposition $(i \ i+1)$.

The $BC_N$ Coxeter group has generating set denoted $S=\{s_0,\ldots,s_{N-1}\}$ with
$$
m_{ij}
= 
\begin{cases}
2, & \text{ if } | i - j | >1,\\
3, & \text{ if } | i-j | =1,\\
1, & \text{ if } i=j.
\end{cases}
\quad \text{ for } 1\leq i,j\leq N-1.
$$
and
$$
m_{0i}=m_{i0} = 
\begin{cases}
2, & \text{ if } i>0,\\
4, & \text{ if } i=1,\\
1, & \text{ if } i=0.
\end{cases}
$$
An explicit presentation of the type $BC$ Coxeter group is the wreath product $S_2 \wr S_N$, which by definition is isomorphic to the semi--direct product $(S_2)^N \rtimes S_N$, where $S_N$ acts on $(S_2)^N$ by permuting the co--ordinates. From this description, it is clear that $\vert W\vert=2^NN!$. Another description is that $W$ is the group of permutations on $\{-N,\ldots,-1,1,\ldots,N\}$ which preserve the pairs $\{-i,i\}$. In this description, the elements of $S$ can be described by the permutations $s_0=(-1\ 1)$ and $s_i = (i \ i+1)(-i \ \ -i-1)$ for $1 \leq i \leq N-1$. 

Every $w\in W$ can be written as a word $s_{i_1}\cdots s_{i_l}$ with letters in $S$; let $l(w)$ be the length of the shortest such word. For $S_N$, viewed as the symmetric group on $N$ letters, this length function is
$$
l(w) = \left| \{ 1\leq  i < j \leq N : w(i)>w(j)\} \right| 
$$
For the $BC_N$ Coxeter group, represented as a subgroup of $S_{2N}$, the length function is 
$$
l(w) = \frac{1}{2} \left| \{ -N \leq i < j \leq N : w(i)>w(j)\} \right| + \frac{1}{2} \left| \{1 \leq i \leq N : w(i)<0\}\right|.
$$

Given a subset $I\subseteq S$, let $W_I$ denote the subgroup of $W$ generated by $I$. The pair $(W_I,I)$ is itself a Coxeter group, and such subgroups are called parabolic subgroups. The length function on $W_I$ is simply the restriction of the length function on $W$. For any parabolic group $H\leq W$, every left coset of $H$ has a unique element of minimal length. Let $D_H\subset W$ denote the set these distinguished left coset representatives. Furthermore, every $w\in W$ can be written uniquely as a product $w=xb$, where $x\in D_H, b\in H$ and $l(w) = l(x)+l(b)$. 

This also extends to double cosets. Given any two parabolic subgroups $H',H\leq W$, every double coset of $H'$ and $H$ has a unique element of minimal length. Let $D_{H',H} = D_{H'}^{-1} \cap D_H$ denote the set of these distinguished double coset representatives. One might expect that every $w\in W$ can be written uniquely as a product $w=axb$, where $a \in H', x \in D_{H',H}, b \in H$ and $l(w) = l(a) + l(w) + l(b)$, but this turns out not to be true. Instead, one must require that $a \in H' \cap D_K$, where $K$ is another parabolic subgroup that depends on $w$, $H'$ and $H$. The details of $K$ will not be needed here.

Given any subset $Y\subset W$, let $Y(q)$ be the Poincar\'{e} series of a Coxeter group $W$, defined by 
$$
Y(q) = \sum_{w \in Y} q^{l(w)}. 
$$ 
We will call a probability measure $\mathbb{P}$ on $Y$ the $q$--exchangeable probability measure if
$$
\mathbb{P}(y) = \frac{q^{l(y)}}{Y(q)}
$$
If $Y=W$, then the normalizing constant has the simple expression
$$
W(q) = \prod_{i=1}^n [e_i+1]_q,
$$
where the $e_i$ are the exponents of $W$ and $[m]_q = 1 + q + \ldots + q^{m-1} = (1-q^m)/(1-q)$ is a $q$--deformed integer. In the case of the $BC_N$ Coxeter group, the exponents are $1,3,\ldots,2N-1$. In the case of the $A_{N-1}$ Coxeter group, the exponents are $1,2,\ldots,N-1$.

\subsection{Interacting Particles}
Let us first describe the state space of the interacting particle system.

Fix two sequences of positive integers $\mathbf{m}=(m_x:x \in \{1,\ldots,L\})$ and $\mathbf{N}=(N_1,\ldots,N_n)$ which both sum to $N$.  Let $\mathcal{S}(\mathbf{N},\mathbf{m})$ denote the set of all particle configurations on the lattice $\{1,\ldots,L\}$ where exactly $m_x$ particles occupy site $x$, subject to the constraint that there are exactly $N_i$ particles of species $i$ or $-i$. More specifically, an element of $\mathcal{S}(\mathbf{N},\mathbf{m})$ can be written as 
$$
\vec{k} = (k_x^{(i)}: x \in \{1,\ldots, L\},  i \in \{-n,\ldots,-1,1,\ldots,n\})
$$
The non--negative integer $k_x^{(i)}$ represents the number of particles of species $i$ at vertex site $x$. We require that 
$$
\sum_{i=-n}^n k_x^{(i)} = m_x, \quad  \sum_{x=1}^L (k_x^{(i)}+k_x^{(-i)}) = N_i
$$
Let $\mathcal{S}^{(0)}(\mathbf{N}, \mathbf{m})$ denote the set of particle configurations, subject to the constraint that there are exactly $N_i$ particles of species $i-1$ or $-(i-1)$. In other words, we allow there to be particles of species $0$, and
$$
\sum_{i=-n}^n k_x^{(i)} = m_x , \quad  \sum_{x=1}^L (k_x^{(i-1)}+k_x^{(-i+1)}) = N_{i}  \text{ for all } 2 \leq i \leq n, \quad \sum_{x=1}^L k_x^{(0)} = N_1.
$$

Let $\mathcal{S}_1(\mathbf{N}, \mathbf{m})$ denote the set of $\vec{k}$ satisfying
$$
k_1^{(i)} = k_1^{(-i)}, \quad \sum_{i=-n}^n k_1^{(i)} = 2m_1, \quad \sum_{i=-n}^n k_x^{(i)}=m_x, \text{ for }  x \geq 2, \quad \quad \frac{1}{2}( k_1^{(i)} + k_1^{(-i)} ) +  \sum_{x=2}^L (k_x^{(i)}+k_x^{(-i)}) = N_{i-1}
$$
In words, every particle configuration of $\mathcal{S}_1(\mathbf{N}, \mathbf{m})$ can be constructed by placing exactly $m_x$ particles at site $x$, with a total of $N_i$ particles of species $i$ or $-i$, and then ``cloning'' every particle of species $i$ at $1$ with a particle of species $-i$. Similarly to before, let $\mathcal{S}_1^{(0)}(\mathbf{N}, \mathbf{m})$ denote the same state space, but allowing particles of species $0$ and fixing exactly $N_i$ particles of species $i-1$ or $-(i-1)$.

\begin{figure}
\begin{center}
\begin{tikzpicture}[scale=0.9, every text node part/.style={align=center}]
\usetikzlibrary{arrows}
\usetikzlibrary{shapes}
\usetikzlibrary{shapes.multipart}

\draw (-1,0)--(3,0);
\draw[dashed] (-0.33,0)--(-0.33,2.5);
\draw (0.33,0)--(0.33,0.66);
\draw (1,0)--(1,0.66);
\draw (1.66,0)--(1.66,0.66);
\draw (2.33,0)--(2.33,0.66);
\draw (3,0)--(3,0.66);
\draw (0.66,0.33) circle (8pt);
\draw (-0.66,0.33) circle (8pt);
\draw (-0.66,1) circle (8pt);
\draw (-0.66,1.66) circle (8pt);
\draw (0,0.33) circle (8pt);
\draw (0,1) circle (8pt);
\draw (0,1.66) circle (8pt);
\draw (0.66,1) circle (8pt);
\draw (1.33,0.33) circle (8pt);
\draw (2,1) circle (8pt);
\draw (2,0.33) circle (8pt);
\draw (2.66,0.33) circle (8pt);

\node at (-0.66,1.66) {$-4$};
\node at (-0.66,1) {$-3$};
\node at (-0.66,0.33) {$3$};
\node at (0,1.66) {$4$};
\node at (0,1) {$3$};
\node at (0,0.33) {$-3$};
\node at (0.66,0.33) {$2$};
\node at (0.66,1)  {$-4$};
\node at (1.33,0.33) {$-1$};
\node at (2,1) {$4$};
\node at (2,0.33) {$1$};
\node at (2.66,0.33) {$2$};
\end{tikzpicture}
\end{center}
\caption{This shows an element of $\mathcal{S}(\mathbf{N},\mathbf{m})$ with $N=9$ and $\mathbf{N}=(2,2,2,3),\mathbf{m}=(3,2,1,2,1)$.}
\end{figure}
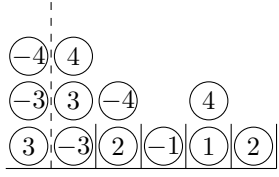

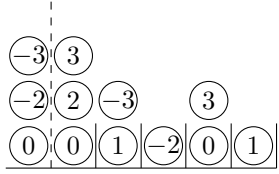
\begin{figure}
\begin{center}
\begin{tikzpicture}[scale=0.9, every text node part/.style={align=center}]
\usetikzlibrary{arrows}
\usetikzlibrary{shapes}
\usetikzlibrary{shapes.multipart}

\draw (-1,0)--(3,0);
\draw[dashed] (-0.33,0)--(-0.33,2.5);
\draw (0.33,0)--(0.33,0.66);
\draw (1,0)--(1,0.66);
\draw (1.66,0)--(1.66,0.66);
\draw (2.33,0)--(2.33,0.66);
\draw (3,0)--(3,0.66);
\draw (0.66,0.33) circle (8pt);
\draw (-0.66,0.33) circle (8pt);
\draw (-0.66,1) circle (8pt);
\draw (-0.66,1.66) circle (8pt);
\draw (0,0.33) circle (8pt);
\draw (0,1) circle (8pt);
\draw (0,1.66) circle (8pt);
\draw (0.66,1) circle (8pt);
\draw (1.33,0.33) circle (8pt);
\draw (2,1) circle (8pt);
\draw (2,0.33) circle (8pt);
\draw (2.66,0.33) circle (8pt);

\node at (-0.66,1.66) {$-3$};
\node at (-0.66,1) {$-2$};
\node at (-0.66,0.33) {$0$};
\node at (0,1.66) {$3$};
\node at (0,1) {$2$};
\node at (0,0.33) {$0$};
\node at (0.66,0.33) {$1$};
\node at (0.66,1)  {$-3$};
\node at (1.33,0.33) {$-2$};
\node at (2,1) {$3$};
\node at (2,0.33) {$0$};
\node at (2.66,0.33) {$1$};
\end{tikzpicture}
\end{center}
\caption{This shows an element of $\mathcal{S}^{(0)}(\mathbf{N},\mathbf{m})$ with $N=9$ and $\mathbf{N}=(2,2,2,3),\mathbf{m}=(3,2,1,2,1)$.}
\end{figure}

\begin{figure}
\begin{center}
\begin{tikzpicture}[scale=0.9, every text node part/.style={align=center}]
\usetikzlibrary{arrows}
\usetikzlibrary{shapes}
\usetikzlibrary{shapes.multipart}

\draw (-1,0)--(3,0);
\draw (-1,0)--(-1,2.5);
\draw (0.33,0)--(0.33,0.66);
\draw (1,0)--(1,0.66);
\draw (1.66,0)--(1.66,0.66);
\draw (2.33,0)--(2.33,0.66);
\draw (3,0)--(3,0.66);
\draw (0.66,0.33) circle (8pt);
\draw (-0.66,0.33) circle (8pt);
\draw (-0.66,1) circle (8pt);
\draw (-0.66,1.66) circle (8pt);
\draw (0,0.33) circle (8pt);
\draw (0,1) circle (8pt);
\draw (0,1.66) circle (8pt);
\draw (0.66,1) circle (8pt);
\draw (1.33,0.33) circle (8pt);
\draw (2,1) circle (8pt);
\draw (2,0.33) circle (8pt);
\draw (2.66,0.33) circle (8pt);

\node at (-0.66,1.66) {$-4$};
\node at (-0.66,1) {$-3$};
\node at (-0.66,0.33) {$3$};
\node at (0,1.66) {$4$};
\node at (0,1) {$3$};
\node at (0,0.33) {$-3$};
\node at (0.66,0.33) {$2$};
\node at (0.66,1)  {$-4$};
\node at (1.33,0.33) {$-1$};
\node at (2,1) {$4$};
\node at (2,0.33) {$1$};
\node at (2.66,0.33) {$2$};
\end{tikzpicture}
\end{center}
\caption{This shows an element of $\mathcal{S}_1(\mathbf{N},\mathbf{m})$ with $N=9$ and $\mathbf{N}=(2,2,2,3),\mathbf{m}=(3,2,1,2,1)$.}
\end{figure}
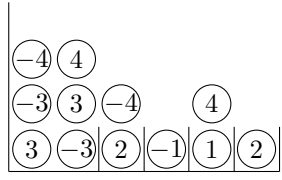

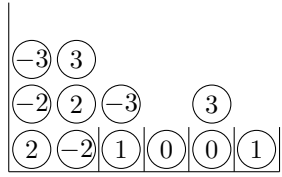
\begin{figure}
\begin{center}
\begin{tikzpicture}[scale=0.9, every text node part/.style={align=center}]
\usetikzlibrary{arrows}
\usetikzlibrary{shapes}
\usetikzlibrary{shapes.multipart}

\draw (-1,0)--(3,0);
\draw (-1,0)--(-1,2.5);
\draw (0.33,0)--(0.33,0.66);
\draw (1,0)--(1,0.66);
\draw (1.66,0)--(1.66,0.66);
\draw (2.33,0)--(2.33,0.66);
\draw (3,0)--(3,0.66);
\draw (0.66,0.33) circle (8pt);
\draw (-0.66,0.33) circle (8pt);
\draw (-0.66,1) circle (8pt);
\draw (-0.66,1.66) circle (8pt);
\draw (0,0.33) circle (8pt);
\draw (0,1) circle (8pt);
\draw (0,1.66) circle (8pt);
\draw (0.66,1) circle (8pt);
\draw (1.33,0.33) circle (8pt);
\draw (2,1) circle (8pt);
\draw (2,0.33) circle (8pt);
\draw (2.66,0.33) circle (8pt);

\node at (-0.66,1.66) {$-3$};
\node at (-0.66,1) {$-2$};
\node at (-0.66,0.33) {$2$};
\node at (0,1.66) {$3$};
\node at (0,1) {$2$};
\node at (0,0.33) {$-2$};
\node at (0.66,0.33) {$1$};
\node at (0.66,1)  {$-3$};
\node at (1.33,0.33) {$0$};
\node at (2,1) {$3$};
\node at (2,0.33) {$0$};
\node at (2.66,0.33) {$1$};
\end{tikzpicture}
\end{center}
\caption{This shows an element of $\mathcal{S}^{(0)}_1(\mathbf{N},\mathbf{m})$ with $N=9$ and $\mathbf{N}=(2,2,2,3),\mathbf{m}=(3,2,1,2,1)$.}
\end{figure}

Now we proceed to defining the dynamics of the multi--species ASEP$(q,\mathbf{m})$, which was introduced in \cite{KIMRN}, generalizing the single--species model in \cite{CGRS}. Suppose that the current state is a particle configuration $\vec{k}$, which can be in any of the four sets described above. The jump rate for a particle of species $i$ at lattice site $x$ to switch places with a particle of species $j<i$ at lattice site $x+1$ is
$$
\frac{q^{ \sum_{s>i} k_x^{(s)}} [ k_x^{(i)}]_q}{[m_x]_q} \cdot \frac{q^{\sum_{r<j} k_{x+1}^{(r)}} [ k_{x+1}^{(j)}]_q}{[m_{x+1}]_q},
$$
and the jump rate for a particle of species $i$ at lattice site $x+1$ to switch places with a particle of species $j<i$ at lattice site $x$ is
$$
q \cdot \frac{  q^{\sum_{s>j} k_x^{(s)} } [ k_x^{(j)} ]_q }{  [m_x]_q }  \cdot \frac{ q^{\sum_{r<i} k_{x+1}^{(r)} }   [ k_{x+1}^{(i)}]_q }{ [m_{x+1}]_q}.
$$

The jump rates of the multi--species ASEP$(q,\vec{m})$ can be related to the jump rates of multi--species ASEP.  Let $\mathbf{1}$ denote the sequence $(1,1,\ldots,1)$. Define a map $\Lambda: \mathcal{S}(\mathbf{N}, \mathbf{m}) \rightarrow \mathcal{S}(\mathbf{N}, \mathbf{1})$ as a product  $\Lambda=\bigotimes_{x\in \mathbb{Z}} \Lambda^{(x)}$, where each $\Lambda^{(x)}$ splits the $m_x$ particles at lattice site $x$ randomly along $m_x$ sites, in such a way that the image of each $\Lambda^{(x)}$ is $q$--exchangeable. Define the map $\Phi:  \mathcal{S}(\mathbf{N}, \mathbf{1})\rightarrow  \mathcal{S}(\mathbf{N}, \mathbf{m})$ by fusing adjacent lattice sites. See Figure \ref{LP}. The maps $\Lambda$ and $\Phi$ can be similar defined on $\mathcal{S}_1, \mathcal{S}^{(0)}, \mathcal{S}_1^{(0)}$. 

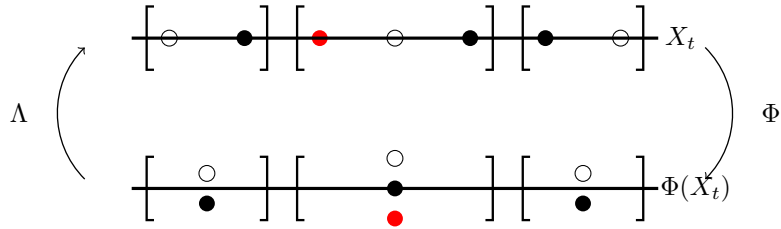
\begin{figure}
\begin{center}
\begin{tikzpicture}
\draw (1,2) circle (3pt);
\draw (4,2) circle (3pt);
\fill[black] (5,2) circle (3pt);
\draw (7,2) circle (3pt);
\fill[black] (2,2) circle (3pt);
\fill[red] (3,2) circle (3pt);
\fill[black] (6,2) circle (3pt);
\draw [very thick](0.5,2) -- (7.5,2);
\draw (8,2) node (a) { $ $};
\draw (8,0) node (b) { $ $};
\draw (0,2) node (c) { $$};
\draw (0,0) node (d) { $$};
\draw (a) edge[out=-45,in=45,->, line width=0.5pt] (b);
\draw (d) edge[out=135,in=-135,->, line width=0.5pt] (c);
\draw (0.7,2) node {\Huge $[$};
\draw (2.3,2) node {\Huge $]$};
\draw (2.7,2) node {\Huge $[$};
\draw (5.3,2) node {\Huge $]$};
\draw (5.7,2) node {\Huge $[$};
\draw (7.3,2) node {\Huge $]$};
\draw (7.8,2) node {$X_t$};
\draw (0.7,0) node  {\Huge $[$};
\draw (2.3,0) node {\Huge $]$};
\draw (2.7,0) node {\Huge $[$};
\draw (5.3,0) node {\Huge $]$};
\draw (5.7,0) node {\Huge $[$};
\draw (7.3,0) node {\Huge $]$};
\draw (8,0) node { $\Phi(X_t)$};
\draw [very thick](0.5,-0) -- (7.5,-0);
\draw (1.5,0.2) circle (3pt);
\fill[black] (4,0) circle (3pt);
\draw (4,0.4) circle (3pt);
\draw (6.5,0.2) circle (3pt);
\fill[black] (1.5,-0.2) circle (3pt);
\fill[red] (4,-0.4) circle (3pt);
\fill[black] (6.5,-0.2) circle (3pt);
\draw (9,1) node {$\Phi$};
\draw (-1,1) node {$\Lambda$};
\end{tikzpicture}
\end{center}
\caption{The maps $\Lambda$ and $\Phi$.}
\label{LP}
\end{figure}

If $L$ denotes the generator of the multispecies ASEP and $\widehat{L}$ denotes the generator of the multi--species ASEP, then 
\begin{equation}\label{ref}
\widehat{L} = \Lambda L \Phi.
\end{equation}
This was observed in Theorem 4.2(ii) of \cite{KuanSF}.


\section{}

\subsection{Particle Configurations as Double Coset Representatives}
We now show how particle configurations can be viewed as double coset representatives. A similar construction was made in \cite{KuanAHP}. First consider the case when $W$ is the $A_{N-1}$ Coxeter system, that is, when $W=S_N$. Fix two sequences of positive integers $\mathbf{m}=(m_x:x \in \{1,\ldots,L\})$ and $\mathbf{N}=(N_1,\ldots,N_n)$ which both sum to $N$. Let $\bar{\mathcal{S}}(\mathbf{N}, \mathbf{m}) \subset {\mathcal{S}}(\mathbf{N}, \mathbf{m})$ be the subset consisting solely of particles with positive species number. 

Let $\mathbf{1}$ denote the sequence $(1,1,\ldots,1)$. There are two natural bijections from $\theta_1,\theta_2$ from $W$ to $\bar{\mathcal{S}}(\mathbf{1},\mathbf{1})$. For each $w\in W$, let $\theta_1(w)$ be the particle configuration where there is a particle of species $w^{-1}(x)$ at lattice site $x$ for $1\leq x \leq N$. Let $\theta_2(w)$ be the particle configuration where there is a particle of species $w(x)$ at lattice site $x$ for $1 \leq x \leq N$.

Recalling that $D_{H',H} = D_{H'}^{-1} \cap D_H$, there are two natural chains of surjections $W\rightarrow D_H \rightarrow D_{H',H}$ and $W\rightarrow D_{H'}^{-1} \rightarrow D_{H',H}$, which map an element to its coset representative. There are also chains of surjections $\bar{\mathcal{S}}(\mathbf{1}, \mathbf{1}) \rightarrow \bar{\mathcal{S}}(\mathbf{1}, \mathbf{N}) \rightarrow \bar{\mathcal{S}}(\mathbf{m}, \mathbf{N})$. The first projection $\Pi_{\mathbf{N}}:  \bar{\mathcal{S}}(\mathbf{1}, \mathbf{1}) \rightarrow \bar{\mathcal{S}}(\mathbf{1}, \mathbf{N}) $ is defined by (here, set $N_0=0$ and $N_{[i,j]} = N_{i}+N_{i+1}+\ldots + N_j$ for convenience)
$$
(\Pi_{\mathbf{N}}(\vec{k}) )_x^{(i)} = k_x^{(N_{[1,i]})} + k_x^{(N_{[1,i]}-1)} + \cdots +k_x^{(N_{[1,i-1]}+1)}, \text{ for } 1\leq i \leq n.
$$
In words, $\Pi_{\mathbf{N}}$ is the projection which ``forgets'' the distinction between particles of species $N_{i-1}+1,\ldots,N_i$. This is also called the ``color--blind'' projection. The second projection is $\Phi_{\mathbf{m}}: \bar{\mathcal{S}}(\mathbf{1}, \mathbf{N}) \rightarrow \bar{\mathcal{S}}(\mathbf{m}, \mathbf{N})$ defined by (again setting $m_{[i,j]} = m_i + \ldots + m_j$)
$$
\Phi_{\mathbf{m}}(\vec{k})_x^{(i)} = k_{m_{[1,x+1] }}^{(i)} + k_{m_{[1,x+1] -1}}^{(i)} + \cdots  k_{m_{[1,x]+1}}^{(i)} .
$$
In words, $\Phi_{\mathbf{m}}$ is the projection which ``fuses'' the vertex sites $m_{[1,x]+1},\ldots,m_{[1,x+1]}$.

These sequences of surjections can be related to each other in the following way:

\begin{prop}\label{Comm}
Set $H=S(\mathbf{N}) = S(N_1)\times \cdots \times S(N_n)\leq S_N$ and $H'= S(\mathbf{m}) = S(m_1)\times \cdots \times S_{m_L}\leq S_N$.

For any $\sigma \in D_H$, the set $\theta_1(\sigma H) \subset \bar{\mathcal{S}}(\mathbf{1}, \mathbf{1})$ is a fiber of $\Pi_{\mathbf{N}}$. Thus, there is a bijection $\theta'_1: D_H \rightarrow \bar{\mathcal{S}}(\mathbf{1}, \mathbf{N})$ defined by $\Pi_{\mathbf{N}}^{-1}(\theta_1'(\sigma) )= \theta_1(\sigma H)$. For any $\tau \in D_{H',H}$, the set $\theta_1'(H'\tau)$ is a fiber of $\Phi_{\mathbf{m}}$, and thus there is a bijection $\theta_1'': D_{H',H} \rightarrow \bar{\mathcal{S}}(\mathbf{m}, \mathbf{N})$ defined by $\Phi_{\mathbf{m}}^{-1}(\theta_1''(\tau)) = \theta_1'(H'\tau)$.

Similarly, for any $\sigma \in D_H^{-1}$, the set $\theta_2(H \sigma) \subset  \bar{\mathcal{S}}(\mathbf{1}, \mathbf{1})$ is a fiber of $\Pi_{\mathbf{N}}$. Thus, there is a bijection $\theta_2': D_H^{-1} \rightarrow \bar{\mathcal{S}}(\mathbf{1}, \mathbf{N})$ defined by $\Pi_{\mathbf{N}}^{-1}(\theta_2'(\sigma) )= \theta_2( H \sigma)$. For any $\tau \in D_{H,H'}$, the set $\theta_2'(\tau H')$ is a fiber of $\Phi_{\mathbf{m}}$, and thus there is a bijection $\theta_2'': D_{H,H'} \rightarrow \bar{\mathcal{S}}(\mathbf{m}, \mathbf{N})$ defined by $\Phi_{\mathbf{m}}^{-1}(\theta_2''(\tau)) = \theta_2'(\tau H')$.

This can be viewed in the commutative diagram
\begin{diagram}
W & \rTo^{\theta_1} & \mathcal{S}(\mathbf{1},\mathbf{1}) & 
\lTo^{\theta_2} & W\\
\dTo &     &  \dTo_{\Pi} & & \dTo\\
D_H & \rTo^{\theta_1'} &  \mathcal{S}(\mathbf{1},\mathbf{N}) & \lTo^{\theta_2'} & D_H^{-1} \\
\dTo &   & \dTo_{\Phi_{\mathbf{m}}} & & \dTo \\
D_{H',H} & \rTo^{\theta_1''}  &  \mathcal{S}(\mathbf{m},\mathbf{N}) & \lTo^{\theta_2''} & D_{H,H'}
\end{diagram}

\end{prop}

\begin{proof}
We only prove the statements about $\theta_1,\theta_1',\theta_1''$, since the proof for $\theta_2,\theta_2',\theta_2''$ are similar. 

By definition, $\theta_1(\sigma H)$ is the set of all configurations where there is a particle of species $b^{-1}\sigma^{-1}(x)$ at lattice site $x$, where $b^{-1}$ ranges over $H=S(\mathbf{N})$. The projection $\Pi$ maps all species $b^{-1}\sigma^{-1}(x)$ to the same species, so $\theta_1(\sigma H)$ is a fiber of $\Pi$. 

Similarly, $\theta_2(H' \tau)$ is the set of all configurations where there is a particle of species $\tau^{-1}a^{-1}(x)$ at lattice site $x$, where $a$ ranges over $H'=S(\mathbf{m})$. This is equivalent to having a particle of species $\tau^{-1}(x)$ at lattice site $a(x)$. Since $\Phi_{\mathbf{m}}$ fuses the lattice sites $a(x)$ together, this shows that $\theta_2(H' \tau)$ is a fiber of $\Phi$.

\end{proof}

We now proceed to an analogous result for the case when $W$ is of type $BC_N$. In that case, there is the additional possibility that $s_0$ is in the parabolic subgroups. Given $\mathbf{N}=(N_1,\ldots,N_n)$ such that $N_1 + \ldots + N_n=N$, let $S(\mathbf{N})=S(N_1) \times \cdots \times S(N_n) \leq S_N \leq W$, and let $S^{(0)}(\mathbf{N})\leq W$ be generated by $S(\mathbf{N})$ and $s_0$.

\begin{prop}\label{P2}
Fix two sequences of integers $\mathbf{N}=(N_1,\ldots,N_n)$ and $\mathbf{m}=(m_1,\ldots,m_L)$, where each sequence sums to $N$. Then there is the commutative diagram of sets

\begin{diagram}
W & \rTo^{\sim} & \mathcal{S}(\mathbf{1},\mathbf{1}) & \lTo^{\sim} & W\\
\dTo &     &  \dTo_{\Pi}& & \dTo \\
D_H & \rTo^{\sim} & B& \lTo^{\sim} & D_H^{-1} \\
\dTo &   & \dTo_{\Phi}& & \dTo  \\
D_{H',H} & \rTo^{\sim}  & A& \lTo^{\sim}& D_{H,H'}
\end{diagram}
(where the $\sim$ indicates a bijection) in the following cases:

(a) When the parabolic subgroups are $H'=S(\mathbf{m})$ and $H=S(\mathbf{N})$, the partition $\Pi$ is $$ \{ \pm\{N_1+\ldots+N_{i-1}+1,\ldots,N_1+ \ldots + N_{i}\} :  1\leq i \leq n  \},$$ and
$
B= \mathcal{S}(\mathbf{N}, \mathbf{1}), A= \mathcal{S}(\mathbf{N}, \mathbf{m}).
$

(b) When the parabolic subgroups are $H'=S(\mathbf{m})$ and $H=S^{(0)}(\mathbf{N})$, the partition $\Pi$ is $$ \{ \pm\{N_1+\ldots+N_{i-1}+1,\ldots,N_1+ \ldots + N_{i}\} :  2\leq i \leq n  \} \coprod \{-N_1,\ldots,-1,1,\ldots,N_1\} ,$$
and
$
B= \mathcal{S}^{(0)}(\mathbf{N}, \mathbf{1}), A= \mathcal{S}^{(0)}(\mathbf{N}, \mathbf{m}).
$

(c) When the parabolic subgroups are $H'=S^{(0)}(\mathbf{m})$ and $H=S(\mathbf{N})$, the partition $\Pi$ is  $$ \{ \pm\{N_1+\ldots+N_{i-1}+1,\ldots,N_1+ \ldots + N_{i}\} :  1\leq i \leq n  \},$$
and
$
B= \mathcal{S}(\mathbf{N}, \mathbf{1}),A= \mathcal{S}_1(\mathbf{N}, \mathbf{m}).
$

(d) When the parabolic subgroups are $H'=S^{(0)}(\mathbf{m})$ and $H=S^{(0)}(\mathbf{N})$, the partition $\Pi$ is  $$ \{ \pm\{N_1+\ldots+N_{i-1}+1,\ldots,N_1+ \ldots + N_{i}\} :  2\leq i \leq n  \} \coprod \{-N_1,\ldots,-1,1,\ldots,N_1\} ,$$
and
$
B= \mathcal{S}^{(0)}(\mathbf{N}, \mathbf{1}), A= \mathcal{S}^{(0)}_1(\mathbf{N}, \mathbf{m}).
$

\end{prop}
\begin{proof}
This is idential to the proof of Proposition \ref{Comm}.
\end{proof}

\subsection{Dynamics and stationary measures}
Define a Markov operator from $\Lambda$ from $D_{H',H}$ to $D_{H}$. Each $w\in D_{H',H}$ will be mapped to a random element in the  coset $wH$. More specifically, 
$$
\mathbb{P}(\Lambda(w) = wb) = \frac{q^{l(b)}}{H(q)} \text{ for all } b\in H(q).
$$
If $W=A_{N-1}$, we know that the  ASEP$(q,\mathbf{m})$ jump rates are given by the formula $\widehat{L}= \Lambda L \Phi$. We will now see the analogous construction for the type $BC_N$ Coxeter group.

Away from the boundary, the evolution on $D_{H',H}$ will be the same as a multi--species ASEP$(q,\mathbf{m})$, but at the boundary the jump rates will differ, depending on whether or not $H'$ contains $s_0$. We will consider two cases separately:

\underline{Case 1:}

First, consider when $s_0 \notin H'$. In this case, $\Lambda$ is the same $\Lambda$ as in the ASEP$(q,\mathbf{m})$. Given particles at lattice site $1$, we consider a ``mirror'' set of particles at lattice site $0$, consisting of the same particles at lattice site $1$, but with the negative of the species numbers. A particle of species $j$ at lattice site $1$ is replaced with a particle of species $-j$ at rate
$$
\begin{cases}
q^{ \sum_{r<j} k_1^{(r)}} \dfrac{ [k_1^{(j)}]_q}{[m_1]_q} , & \text{ if } j<0,\\
\ &\\
q \cdot q^{ \sum_{r<j} k_1^{(r)}} \dfrac{[k_1^{(j)}]_q}{[m_1]_q}  , & \text{ if } j>0.
\end{cases}
$$
These are essentially the jump rates of the ASEP$(q,\mathbf{m})$ with no contribution from the lattice site $0$. 

\underline{Case 2:}

Now suppose $s_0 \in  H'$. In this case, particles cannot enter or exit the lattice at the left boundary. The system evolves as an ASEP$(q,\mathbf{m}')$, where $\mathbf{m}' = (2m_1,m_2,\ldots,m_L)$, with one key difference: in order to maintain that the particles at lattice site $1$ consist of pairs of species $(j,-j)$, we mandate that when a particle of species $j$ at $1$ and a species of particle $i$ at $2$ switch places, the particle at species $-j$ at $1$ is replaced with a particle of species $-i$ instantaneously. 

\begin{prop}
The  $q$--exchangeable measures are stationary under the two processes defined above.
\end{prop}
\begin{proof}
By construction, the generator of the process can be defined as $\Lambda L \Phi$, where $L$ is the generator of multi--species ASEP. Since $\Lambda,L$ and $\Phi$ all preserve $q$--exchangeability, this means that $q$--exchangeable measures are stationary. 
\end{proof}

\subsection{Color Position Symmetry}
For any Coxeter system $(W,S)$, let $\mathbb{C}[W]$ denote the group algebra of $W$. For any $s\in S$, define the linear map $L_{s,x}$ on $\mathbb{C}[W]$ by 
$$
L_{s, x}(w)=\left\{\begin{array}{ll}{(1-x) w+x s w}, & { l(s w)>l(w)} \\ {(1-q x) w+q x s w,} & {   l(s w)<l(w) }\end{array}\right.
$$
Fix an arbitrary set of elements $s_{i_1},\ldots,s_{i_n}$ and parameters $x_1,\ldots,x_n$ Define the coefficients $f_n(w \rightarrow \pi)$ by 
$$
L_{s_{i_n},x_n}\cdots L_{s_{i_1},x_1} w = \sum_{\pi \in W} f_n(w\rightarrow \pi) \pi.
$$
Similarly, define the coefficients $\tilde{f}_n(w \rightarrow \pi)$ by 
$$
L_{s_{i_1},x_1}\cdots L_{s_{i_n},x_n} w = \sum_{\pi \in W} \tilde{f}_n(w\rightarrow \pi) \pi.
$$
Note that the ordering of the $t_j,x_j$ are reversed.
In \cite{BorodinBufetovCP}, it is shown that for all $\pi \in S_N$,
\begin{equation}\label{ColPos}
f_n(e\rightarrow \pi) = \tilde{f}_n(e\rightarrow \pi^{-1}).
\end{equation}
In the context of ASEP, \eqref{ColPos} can be viewed as a color--position symmetry. 

As seen above, the color and position permutations are given an algebraic interpretation. Namely, a permutation of the colors is viewed as the \textit{left} action of $S_N$ on itself, while a permutation of the positions is viewed as the \textit{right} action of $S_N$ on itself. In light of this interpretation, it makes sense to view $w_{s,x}w$ as a \textit{left} action, with $w\tilde{L}_{s,x}$ as its corresponding \textit{right} action (the definition of $\tilde{L}_{s,x}$ will be given later).

Here are some heuristics to see the color--position symmetry through left and right actions. If 
\begin{equation}\label{E}
L_{s,x} e = e \tilde{L}_{s,x},
\end{equation}
then we would expect
\begin{equation}\label{H}
L_{s_{i_n},x_n}\cdots L_{s_{i_1},x_1} e = e \tilde{L}_{s_{i_n},x_n}\cdots \tilde{L}_{s_{i_1},x_1} 
\end{equation}
Note that the $s_{i_k},x_k$ are applied to $e$ in reverse orders on both sides of \eqref{H}. Equation \eqref{H} could then be used to prove \eqref{ColPos}. However, this heuristic implicitly uses the associativity of the left and right actions:
\begin{equation}\label{F}
(L_{s_i,x_1} w) \tilde{L}_{s_j,x_2} = L_{s_i,x_1} (w \tilde{L}_{s_j,x_2} ),
\end{equation}
which is not immediately obvious. For instance, it would be false if $q$ were allowed to take different values.

Now define the right action
$$
w \tilde{L}_{s, x}=\left\{\begin{array}{ll}{(1-x) w+x w s,} & {l(w s)>l(w)} ,\\ {(1-q x) w+q x w s,} & {\text { else. }}\end{array}\right.
$$
Note that \eqref{E} holds immediately.
Define the coefficients $\tilde{g}_n(w \rightarrow \pi)$ by 
$$
w \tilde{L}_{s_{i_n},x_n}\cdots \tilde{L}_{s_{i_1},x_1}  =  \sum_{\pi \in W} \tilde{g}_n(w\rightarrow \pi) \pi.
$$

\begin{theorem}\label{L1}
The coefficients $\tilde{f}_n$ and $\tilde{g}_n$ are related via
$$
\tilde{g}_n(e\rightarrow \pi) = \tilde{f}_n(e\rightarrow \pi^{-1}).
$$
If \eqref{H} holds, then so does \eqref{ColPos}.
\end{theorem}
\begin{proof}
The definition of $\tilde{g}_n$ is the same as the definition of $\tilde{f}_n$, except that the multiplication is applied on the right instead of on the left. Thus $\tilde{g}_n(e\rightarrow \pi) = \tilde{f}_n(e\rightarrow \pi^{-1})$.

If \eqref{H} holds, then $f_n(e\rightarrow \pi) = \tilde{g}_n(e\rightarrow \pi)$, which implies \eqref{ColPos}.
\end{proof}

For two transpositions $s_i,s_j$, we say that $(s_i,s_j)$--\textit{associativity} holds for $w \in S_N$ if for any $x_1,x_2$,
$$
(L_{s_1,x_1}w)\tilde{L}_{s_2,x_2} = L_{s_1,x_1}(w \tilde{L}_{s_2,x_2})
$$ 
\begin{lemma}\label{L2}
For any $s_i,s_j \in S$ and any parameters $x_1,x_2$, we have that $(s_i,s_j)$--\textit{associativity} holds for $w$
for the following three cases:
\begin{align*}
l(s_iw)> l(w), \quad l(w s_j)>l(w), &\quad l(s_iw s_j)>l(w s_j), \quad l(s_iw s_j)>l(s_iw ),
\\
l(s_iw)< l(w), \quad l(w s_j)>l(w), &\quad l(s_iw s_j)<l(w s_j), \quad l(s_iw s_j)>l(s_iw ),
\\
l(s_iw)> l(w), \quad l(w s_j)<l(w), &\quad l(s_iw s_j)>l(w s_j), \quad l(s_iw s_j)<l(s_iw )\\
l(s_iw)< l(w), \quad l(w s_j)<l(w), &\quad l(s_iw s_j)<l(w s_j), \quad l(s_iw s_j)<l(s_iw )
\end{align*}
\end{lemma}
\begin{proof}
In the first case, it can be checked that both sides equal
$$
(1-x_1)(1-x_2)w + (1-x_1)x_2 w s_j + x_1(1-x_2)s_iw + x_1x_2 s_i w s_j.
$$
In the second case, it can be checked that both sides equal
$$
(1-qx_1)(1-x_2)w + (1-qx_1)x_2w s_j + qx_1(1-x_2)s_i w + qx_1x_2 s_iw s_j.
$$
In the third case, it can be checked that both sides equal
$$
\left(1-x_{1}\right)\left(1-q x_{2}\right) w+\left(1-x_{1}\right)  q x_{2} w s_j
+x_{1}\left(1-q x_{2}\right) s_i w+q x_{1} x_{2} s_i w s_j.
$$
In the fourth case, it can be checked that both sides equal
$$
(1-qx_1)(1-qx_2)w + (1-qx_1)qx_2 w s_j + qx_1(1-qx_2)s_iw + q^2x_1x_2 s_i w s_j.
$$
\end{proof}

A priori, there are $16$ cases that need to be checked. However, if $l(s_iw)>l(w)$ and $l(w s_j)>l(w)$, then $l(s_iw) = l(w s_j) = l(w)+1$, which means that either $l(s_iw)=l(w s_j)>l(s_iw s_j)$ or $l(s_iw)=l(w s_j)<l(s_iw s_j)$. Similarly, if $l(s_iw)<l(w)$ and $l(w s_j)<l(w)$ then either $l(s_iw)=l(w s_j)<l(s_iw s_j)$ or $l(s_iw)=l(w s_j)>l(s_iw s_j)$. If $l(s_iw)>l(w)$ and $l(w s_j)<l(w)$, then $l(s_iw)=l(w)+1$ and $l(w s_j)=l(w)-1$, which means that $l(s_iw s_j)=l(w)$, forcing $l(s_iw s_j) > l (w s_j)$ and $l(s_iw s_j) < l(s_iw)$. Similarly, $l(s_iw)<l(w)$ and $l(w s_j)>l(w)$ imply that $l(s_iw s_j) < l(w s_j)$ and $l(s_i w s_j)> l(s_iw)$. Thus, there are only six cases total, with four of them checked in the lemma. The remaining two are:
\begin{align*}
l(s_iw)> l(w), \quad l(w s_j)>l(w), &\quad l(s_iw s_j)<l(w s_j), \quad l(s_iw s_j)<l(s_iw ),
\\
l(s_iw)< l(w), \quad l(w s_j)<l(w), &\quad l(s_iw s_j)>l(w s_j), \quad l(s_iw s_j)>l(s_iw ).
\end{align*}
Examples of $w$ in each of the two cases are (respectively) $s_js_i$, and $s_is_js_i$, where $s_i= s_j = (j\ j+1)$ and $s_j=s_{j+1}=(j+1 \ j+2).$ One can check, however, that $(L_{s_1,x_1} w) \tilde{L}_{s_2,x_2} = L_{s_1,x_1} (w \tilde{L}_{s_2,x_2} )$ does not hold for $w = s_is_js_i$. 

Let us refer to the four cases in Lemma \ref{L2} as cases $1,2,3,4$ respectively, and the remaining two cases as cases $5,6$. The proof for case $5$ requires a slightly more delicate proof. In particular, we recall the exchange condition: if $t$ is a transposition and $w=s_i\cdots t_r$ is an arbitrary element of $S_N$ such that $l(tw) < l(w)$, then there exists $h\in \{1,\ldots,r\}$ such that
$t s_i\ldots t_{h-1} = s_i\ldots t_h$. As a consequence, $tw = s_i\ldots \hat{t}_h \ldots t_r$, where the $\hat{t}_h$ indicates that $t_h$ has been deleted. Similarly, if $l(w t)<l(w)$ then $w t= s_i\cdots \hat{t}_k \cdots t_r$ for some $1\leq k \leq r$.

\begin{lemma} Fix $s_i,s_j$ and suppose that $w$ falls into case $5$. Then 
\begin{equation}\label{Same}
\begin{aligned}
( L_{s_1,x} w )\tilde{L}_{s_2,y} &=(1-x)(1-y) w+(1-x) y w s_j  {+x(1-q y) s_i w+x q y s_i w s_j}\\
L_{s_1,x}(w \tilde{L}_{s_2,y}) &=(1-x)(1-y) w+x(1-y) s_i w  {+(1-q x) y w s_j+x q y s_i w s_j}
\end{aligned}
\end{equation}
If $w$ falls into case $6$, then
\begin{equation}\label{Same2}
\begin{aligned}
( L_{s_1,x} w )\tilde{L}_{s_2,y} &= (1-qx_1)(1-qx_2)w + (1-qx_1)qx_2w s_j + qx_1(1-x_2)s_iw + x_2s_iw s_j\\
L_{s_1,x}(w \tilde{L}_{s_2,y}) &= (1-qx_1)(1-qx_2)w + qx_1(1-qx_2)s_iw + q(1-x_1)x_2w s_j + qx_1x_2s_iw s_j.
\end{aligned}
\end{equation}
\end{lemma}
\begin{proof}
This follows from a direct calculation.
\end{proof}

A priori, the two terms in \eqref{Same} are not equal to each other. It turns out, however, that they actually are  equal.

\begin{lemma}
Fix $s_i,s_j$ and suppose that $w$ falls into case $5$ or case $6$. Then $s_i w  = w s_j$,  and furthermore
\begin{align*}
( L_{s_1,x}w) \tilde{L}_{s_2,y} &= L_{s_1,x}(w \tilde{L}_{s_2,y})
\end{align*}
\end{lemma}
\begin{proof}
Suppose that case 5 holds; then by assumption $l(s_i w) > l(w)$ and $l(s_iw s_j)<l(w s_j)$. Write $w = s_{i_1}\ldots s_{i_r}$. By the exchange condition, $s_is_{i_1}\ldots s_{i_r}s_j$ equals either $s_{i_i} \cdots \hat{s}_{i_h} \cdots s_{i_r}s_j$ (for some $h\in \{1,\ldots,r\}$) or equals $s_{s_i}\cdots s_{i_r}$. However, the former situation cannot hold because by assumption $l(s_iw) > l(w)$. Thus $s_iw s_j = w$, or equivalently $s_i w =w s_j$. By \eqref{Same} and by
\begin{equation}\label{K}
(1-x)y + x(1-qy) = x(1-y) + (1-qx)y,
\end{equation}
 the identity holds. 
 
 Now suppose that case $6$ holds: then by assumption $l(s_iw) < l(w)$ and $l(w s_j) < l(s_iw s_j)$, and $l(s_iw s_j) = l(w)$. Let $s_{i_1} \cdots s_{i_l}$ be a reduced expression for $w$. By the exchange condition, there exists $h \in \{1,\ldots, l\}$ such that 
 \begin{equation}\label{Temp}
 s_i s_{i_1}\cdots s_{i_{h-1}} = s_{i_1} \cdots s_{i_h}.
 \end{equation}
 Set $\tilde{w} = s_iw s_j$ and let $s_{j_1}\cdots s_{j_l}$ be the reduced expression for $\tilde{w}$ (by assumption, $l(w)=l(\tilde{w})$) given by 
$$
j_1 = i_1, \ldots , j_{h-1} = i_{h-1} , j_h = i_{h+1}, \ldots, j_{l-1}=i_{l}, s_{j_l}=s_j.
$$
This is a reduced expression for $\tilde{w}$ because of \eqref{Temp}.

Note that $l(s_i\tilde{w}) = l(w s_j) < l(s_iw s_j)=l(\tilde{w})$. By the exchange condition, either $s_i \tilde{w} $ equals  the permutation $s_{i_1} \cdots \hat{s}_{i_k} \cdots \hat{s}_{i_h} \cdots {s}_{i_{l}}  s_j$ for some $h,k$, or $s_i\tilde{w} = s_{i_1} \cdots \hat{s}_{i_h} \cdots s_{i_{l}}$. The former equality would imply that $l(s_i\tilde{w}s_j) = l(w) = l-2$, which contradicts $l(w)=l$. Therefore the latter equality holds, which then implies $\tilde{w} = s_i s_{i_1} \cdots \hat{s}_{i_h} \cdots s_{i_{l}}$. By \eqref{Temp}, this implies that $\tilde{w} = s_{i_1}\cdots s_{i_l} = w$, which means that $s_iw = w s_j$, as claimed. The identity then holds by \eqref{Same2} and \eqref{K}.
\end{proof}

\begin{remark}
In the language of Hecke algebras, the color--position symmetry can be viewed as the commutativity of left and right actions; see, for example, the proposition in section 7.2 of \cite{HumphreysBook}. Note that the approach of \cite{Bufetov2020} describes color--position symmetry in terms of an involution on the Hecke algebra, rather than in terms of left and right actions. 
\end{remark}

As stated, Theorem \ref{L1} is purely an algebraic statement. It can also be stated in terms of interacting particles. Before doing so, we first describe a coupling between the multi--species ASEP$(q,j)$ and its time reversal. This involves the graphical representation of an interacting particle system (see e.g. \cite{LiggBook}). In the simplest case of the single--species ASEP, we have independent Poisson processes $\{ \mathcal{P}(x)\}_{x\in \mathbb{Z}}$, each on state space $\mathbb{R}_{\geq 0}$ with rates $c(x)=1$. At each point from $\mathcal{P}(x)$, apply the Markov operator $L_{(x\ x+1),1}$; this corresponds to an update of the particle configuration. With this description, an ASEP during times $t\in [0,T]$ can be coupled with a time--reversed ASEP by mapping every $\mathcal{P}(x)$ to $T-\mathcal{P}(x)$. Note that this coupling works for any initial conditions on ASEP and its time reversal. 

The coupling between multi--species ASEP$(q,\vec{m})$ and its time reversal is similar; instead of applying the Markov operator $w_{(x \ x+1),1}$, randomly swap a particle at $x$ with a particle at $x+1$ that is consistent with the multi--species ASEP$(q,\vec{m})$. By \eqref{ref}, this random swap is equal to $\Lambda s_{(x\ x+1),1}\Phi$.

\begin{theorem}
Fix $H=S(\mathbf{N})$ and $H'=S(\mathbf{m})$. Let $\vec{k}_t \in \mathcal{S}(\mathbf{N}, \mathbf{m})$ evolve as a multi--species ASEP$(q,\vec{m})$ with initial condition $\vec{k}_0$, and set $w_t \in D_{H',H}$ be defined by $w_t=(\theta_1'')^{-1}(\vec{k}_t)$. Let $\vec{l}_t \in \mathcal{S}(\mathbf{N}, \mathbf{m})$ evolve as a time--reversed multi--species ASEP$(q,\vec{m})$ with initial condition $\theta_2''(e)$. Let $v_t \in D_{H,H'}$ be defined by $v_t=(\theta_2'')^{-1}(\vec{k}_t)$. 

Then, at any time $t$, the distribution of $w_t$ is the same as the distribution of $(w_0^{-1}v_t)^{-1}$.
\end{theorem}
\begin{proof}
Because this proof is similar to Theorem 3.1 of \cite{BorodinBufetovCP}, which is itself similar the proof of Theorem 1.4 of \cite{AAV}, we keep the proof short. 

Let $\iota:W\rightarrow W$ map every $w$ to its inverse $w^{-1}$. Use the same symbol to denote the restriction to any subset of $W$. By Theorems \ref{L1} and Proposition \ref{Comm} (Proposition \ref{P2} for open boundary conditions), this follows once we show 
$$
\iota \Lambda = \Lambda \iota, \quad \iota L_{s,x} = \tilde{L}_{s,x} \iota, \quad \iota \Phi = \Phi \iota.
$$
This follows immediately from the definitions. 
\end{proof}

\subsection{Asymptotic Application}
Before finding asymptotic applications of color--position symmetry, we first need hydrodynamics and local statistics for the single--species ASEP$(q,\mathbf{m})$. The key idea is to use a duality for ASEP$(q,j)$ and previously known results for ASEP. First we recall the results for ASEP, which go back to \cite{BF87} and \cite{AV87}. 

Define the density profile
$$
\rho(x,t) = \mathbb{E}[ k_x(t)],
$$
which by definition takes values in $[0,m_x]$. Here, we consider the \textit{step initial condition}, where $k_x=m_x 1_{x \leq 0}$.

\begin{theorem}\label{BF}
Take all $m_x=1$ (i.e. the ASEP case), and let $\{ x(t): t\geq 0\}$ be a collection of integers such that $\lim_{t\rightarrow \infty} x(t)/t=y \in \mathbb{R}$. Then
$$
\lim_{t\rightarrow \infty} \rho_{\text{step}}(x(t),t) = 
d(y)
:= 
\begin{cases}
0, & y\geq 1-q\\
\tfrac{1}{2}(1- \frac{y}{1-q}), & -(1-q)<y<1-q,\\
1, & y \leq -(1-q).
\end{cases}
$$
\end{theorem}

Now let us turn to the analog of these results for ASEP$(q,j)$. 

\begin{prop}\label{hyd}
Take all values of $m_x$ to equal a fixed $m$. Let $\{x(t): t \geq 0\}$ be a collection of integers such that $\lim_{t\rightarrow \infty} x(t)/t = y \in \mathbb{R}$. Then
$$
\lim_{t\rightarrow \infty} \rho_{\text{step}}(x(t),t) = m d(y),
$$
where $d(y)$ is the function in Theorem \ref{BF}. 
\end{prop}
\begin{proof}
We recall some results and notation from \cite{CGRS}. Given a particle configuration $\eta$, let
$$
N_i(\eta) = \sum_{x \geq i} \eta_x.
$$
Let $\vert \eta\vert$ denote the number of particles in $\eta$ (possibly infinite). Let $x(t)$ denote a single random walker evolving under ASEP$(q,m/2)$, which has the same distribution as letting $\tilde{x}(t)$ evolve as a single random walker under ASEP$(q^m,1/2)$. Then, according to Lemma 3.1 of \cite{CGRS}, there is the duality result
$$
\mathbb{E}_{\eta}[q^{2N_i(\eta(t))}] = q^{2\vert \eta \vert} - \sum_{k=-\infty}^{i-1} q^{-2mk} \mathbb{E}_k\left[ q^{2mx(t)}(1-q^{2\eta_{x(t)}}) q^{2N_{x{(t)}}(\eta)}\right],
$$
where $\eta(t)$ evolves as ASEP$(q,m/2)$.
Then, for step initial conditions $\eta^{\text{step}}$,
\begin{align*}
\mathbb{E}_{\eta^{\text{step}}}[q^{2N_i(\eta(t))}] &= - \sum_{k=-\infty}^{i-1} q^{-2mk} \mathbb{E}_k\left[ q^{2mx(t)}(1-q^{2m}) q^{2N_{x{(t)}}(\eta)} 1_{x(t)<0}\right]\\
&= - \sum_{k=-\infty}^{i-1} q^{-2mk} \mathbb{E}_k\left[ (1-q^{2m})  1_{x(t)<0}\right]\\
&=  (q^{2m}-1) \sum_{k=-\infty}^{i-1} q^{-2mk} \mathbb{P}_k(x(t)<0))
\end{align*}
Since $x(t)$ has the same evolution as ASEP$(q^m,1/2)$, we see that the parameter $m$ only affects the quantity $\mathbb{E}_{\eta^{\text{step}}}[q^{2N_i(\eta(t))}]$ by replacing $q$ with $q^m$. Thus, $N_i(\eta(t)) = m N_i(\tilde{\eta}(t))$, where $\tilde{\eta}(t)$ evolves as the usual ASEP. Combined with Theorem \ref{BF}, this shows that $\rho_{\text{step}}(x(t),t) = md(y)$. 

\end{proof}

We have the following applications to the distribution of a second class particle in the ASEP$(q,m/2)$. Consider a deformed step initial condition, where there are $m$ first class particles at lattice sites to the left at $0$, and no particles to the right of $0$, and $m$ second class particles at $0$. Let $\mathsf{f}^{(1)}_2(t)\leq \ldots \leq \mathsf{f}^{(m)}_2(t)$ denote the positions of the second class particles at time $t$, and as before let $\rho(x,t)$ denote the density profile (of the first class particles).

\begin{theorem} For any lattice site $x$ and any time $t$,
$$
\mathbb{E}[ | \{j: \mathsf{f}_2^{(j)}(t)\leq x\}| ]= \rho(0,t),
$$
where the process on the right--hand--side begins with the shifted step initial condition $\vec{k}_y = 1_{y \leq x} m$.
\end{theorem}
\begin{proof}
Once color--position symmetry is proven, this is similar to the proof of Theorem 4.1 of \cite{BorodinBufetovCP}. The process on the left--hand--side can be coupled with a process which has initial condition consisting of a particles of species $m(x-1)+1,m(x-1)+2,\ldots,mx$ at lattice site $x$. The probability on the left--hand--side is then the probability of the species $0$ particle being located at a site which is $\leq x$. By color--position symmetry, this is the probability of a particles of species $\leq x$ being located at $0$. 
\end{proof}

Combining this with Proposition \ref{hyd}, we have:
\begin{corollary}
As $t\rightarrow \infty$,
$$
\mathbb{E}[ | \{j: \mathsf{f}_2^{(j)}(t)\leq yt\}| ] \rightarrow 
\begin{cases}
m, & y \geq 1-q,\\
\frac{m}{2}( \frac{y}{1-q} -1), & -(1-q) < y < 1-q,\\
0, & y \leq -(1-q).
\end{cases}
$$
\end{corollary}

We note that this corollary is motivated by Theorem 5.2 of \cite{BorodinBufetovCP}.

\bibliographystyle{alpha}
\bibliography{ColorPositionSymmetryv3}

\end{document}